\documentclass[11pt,english]{article}

\usepackage[T1]{fontenc}
\usepackage[latin9]{inputenc}
\usepackage{geometry}               

\usepackage{graphicx}               
\usepackage{float}                  
\usepackage{subfigure}              
\usepackage{algorithm}
\usepackage[noend]{algorithmic}     

\usepackage{amsthm,amsmath,amssymb} 
\usepackage{color}                  
\usepackage{textcomp}               
\usepackage{sectsty}                
\usepackage[normalem]{ulem}         
\usepackage{soul}                   
\usepackage{enumerate}              
\usepackage{microtype}              

\usepackage[bookmarksnumbered=true]{hyperref}   

\graphicspath{{./figures/}}

\theoremstyle{plain}
\newtheorem{theorem}{Theorem}

\theoremstyle{definition}
\newtheorem{definition}{Definition}
\newtheorem{problem}{Problem}

\theoremstyle{remark}
\newtheorem{remark}{Remark}

\def\reals{{\mathbb R}}

\newcommand{\frechet}{Fr\'echet}
\newcommand{\df}{d_F}
\newcommand{\dfd}{d_{dF}}
\newcommand{\npc}{\textbf{NP}-complete}

\date{}

\newcommand{\old}[1]{{{}}}

\begin{document}

\title{On the Chain Pair Simplification Problem}

\author{
Chenglin Fan
\thanks{Montana State University, Bozeman, MT, 59717-3880 USA; \texttt{chenglin.fan@msu.montana.edu}}
\and
Omrit Filtser
\thanks{Ben-Gurion University of the Negev, Beer-Sheva 84105, Israel; \texttt{omritna@post.bgu.ac.il}}
\and
Matthew J. Katz
\thanks{Ben-Gurion University of the Negev, Beer-Sheva 84105, Israel; \texttt{matya@cs.bgu.ac.il}}
\and
Tim Wylie
\thanks{The University of Texas-Pan American, Edinburg, TX, 78539 USA; \texttt{wylietr@utpa.edu }}
\and
Binhai Zhu
\thanks{Montana State University, Bozeman, MT, 59717-3880 USA; \texttt{bhz@cs.montana.edu}}
}

\maketitle

\begin{abstract}

The problem of efficiently computing and visualizing the structural resemblance between a pair of
protein backbones in 3D has led Bereg et al.~\cite{BeregJWYZ08} to pose the Chain Pair Simplification problem (CPS).
In this problem, given two polygonal chains $A$ and $B$ of lengths $m$ and $n$, respectively, one needs to
simplify them simultaneously, such that each of the resulting simplified chains, $A'$ and $B'$, is of length at most $k$ and
the discrete \frechet\ distance between $A'$ and $B'$ is at most $\delta$, where $k$ and $\delta$ are given parameters.

In this paper we study the complexity of CPS under the discrete \frechet\ distance (CPS-3F), i.e.,
where the quality of the simplifications is also measured by the discrete \frechet\ distance.
Since CPS-3F was posed in 2008, its complexity has remained open. However, it was believed
to be \npc, since CPS under the Hausdorff distance (CPS-2H) was shown to be \npc.
We first prove that the weighted version of CPS-3F is indeed weakly \npc\, even on the line, based on a reduction from the set partition problem.
Then, we prove that CPS-3F is actually polynomially solvable, by presenting an $O(m^2n^2\min\{m,n\})$ time algorithm
for the corresponding minimization problem.
In fact, we prove a stronger statement, implying, for example, that if weights are assigned to the vertices of only one of the chains, then the problem remains polynomially solvable.
We also study a few less rigid variants of CPS and present efficient solutions for them.

Finally, we present some experimental results that suggest that (the minimization version of) CPS-3F is significantly better than previous algorithms
for the motivating biological application.

\end{abstract}

\section{Introduction}
Polygonal curves play an important role in many applied areas, such as 3D modeling in computer vision,
map matching in GIS, and protein backbone structural alignment and comparison in computational biology.
Many different methods exist to compare curves in these (and in many other) applications, where one of the more prevalent methods
is the \frechet\ distance~\cite{Frechet1906}.

\old{
The study of polygonal curves impacts many applied areas such as 3D modeling in computer vision,
GIS applications with map matching and routing, wireless networking coverage, and computational biology with protein backbone structure alignment and comparison.
There are many methods used to compare curves for these applications, and one of the most prevalent
is the \frechet\ distance \cite{Frechet1906}.
}

The \emph{\frechet\ distance} is often described by an analogy of a man and a dog connected by a leash, each walking along a curve
from its starting point to its end point. Both the man and the dog can control their speed but they are not allowed to backtrack.
The \frechet\ distance between the two curves is the minimum length of a leash that is sufficient for traversing both curves in this manner.

The \emph{discrete \frechet\ distance} is a simpler version, where, instead of continuous curves, we are given finite sequences of points,
obtained, e.g., by sampling the continuous curves, or corresponding to the vertices of polygonal chains.
Now, the man and the dog only hop monotonically along the sequences of points.
The discrete \frechet\ distance is considered a good approximation of the continuous distance.

One promising application of the discrete \frechet\ distance has been protein backbone comparison.
Within structural biology, polygonal curve alignment and comparison is a central problem in
relation to proteins. Proteins are usually studied using RMSD (Root Mean Square Deviation),
but recently the discrete \frechet\ distance was used to align and compare protein backbones,
which yielded favorable results in many instances~\cite{JiangXZ08,WylieLZ11}.
In this application, the discrete version of the \frechet\ distance makes more sense, because by using it the
alignment is done with respect to the vertices of the chains, which represent $\alpha$-carbon atoms.
Applying the continuous \frechet\ distance will result in mapping of arbitrary points,
which is not meaningful biologically.

There may be as many as 500$\sim$600 $\alpha$-carbon atoms along a protein backbone,
which are the nodes (i.e., points) of our chain. This makes efficient computation essential, and is
one of the reasons for considering simplification.
In general, given a chain $A$ of $n$ vertices, a simplification of $A$ is a chain $A'$ such that $A'$ is ``close'' to $A$ and the number of vertices in $A'$ is significantly less than $n$.
The problem of simplifying a 3D polygonal chains under the discrete \frechet\ distance was first addressed by Bereg et al.~\cite{BeregJWYZ08}.

\begin{figure*}[htp]
  \begin{center}
    \subfigure[Simplifying the chains independently does not necessarily preserve the resemblance between them.] {\includegraphics[scale=0.40]{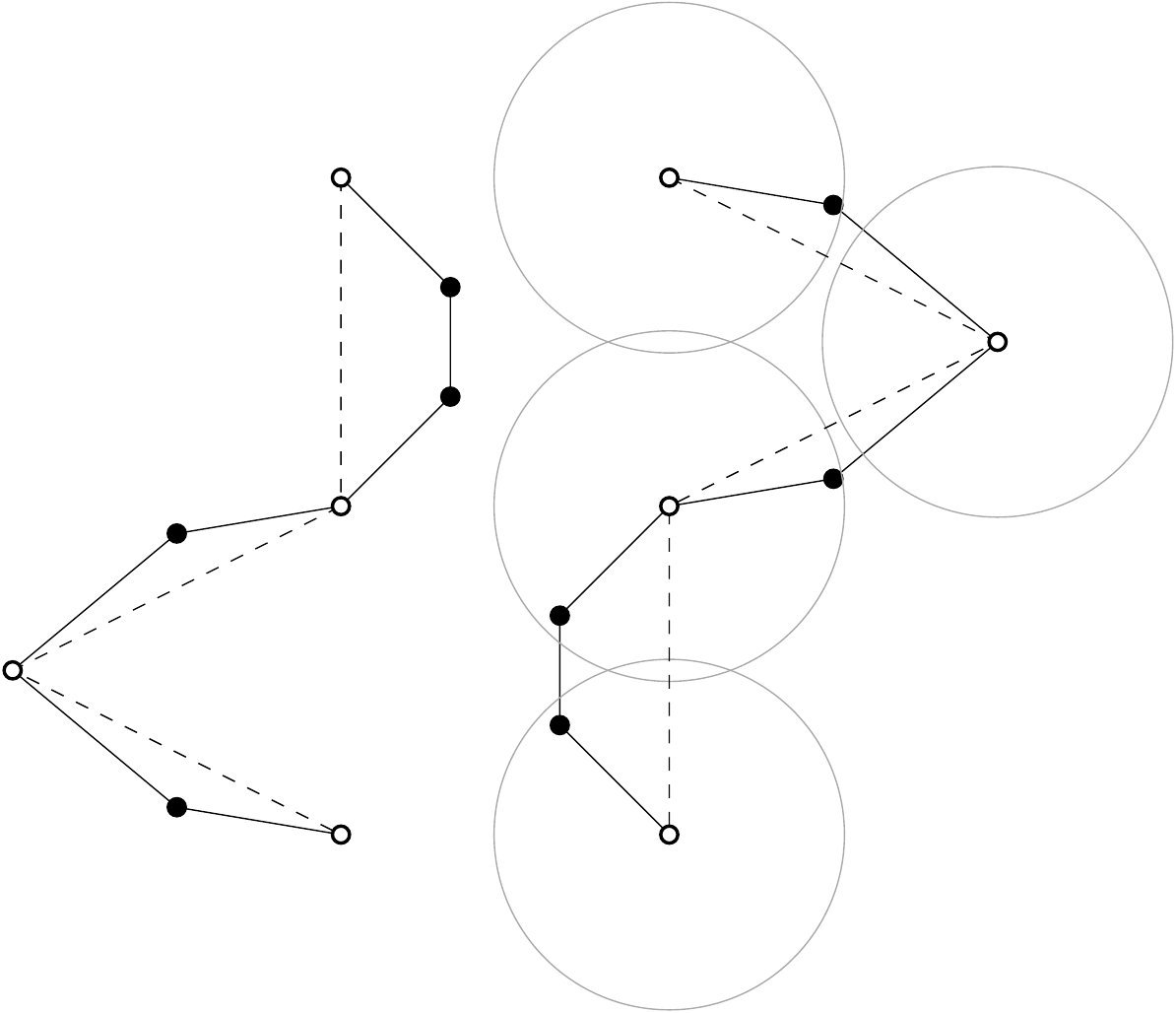}}  \hspace{0.5cm}
    \subfigure[A simplification of both chains that preserves the resemblance between them.] {\includegraphics[scale=0.40]{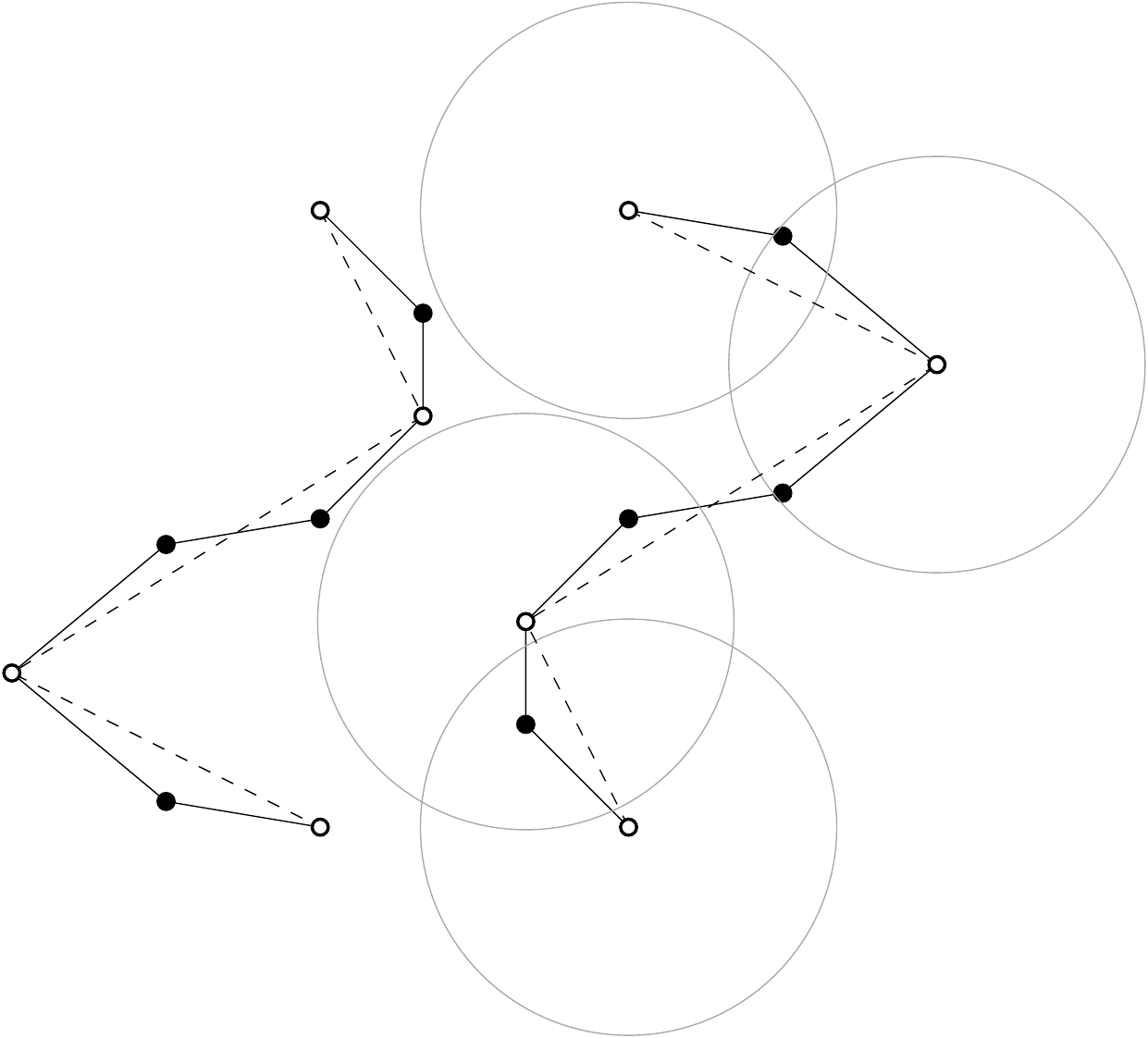}}
  \end{center}
  \caption{Independent simplification vs. simultaneous simplification. Each chain simplification consists of 4 vertices (marked by empty circles) chosen from the corresponding chain. The unit disks illustrate the \frechet\ distance between the right chain in each of the figures and its corresponding simplification; their radius in (b) is larger.}
  \label{fig:separate}
\end{figure*}

Simplifying two aligned chains independently does not necessarily preserve the resemblance between the chains; see Figure \ref{fig:separate}.
Thus, the following question arises: Is it possible to simplify both chains in a way that will retain the resemblance between them?
This question has led Bereg et al.~\cite{BeregJWYZ08} to pose the Chain Pair Simplification problem (CPS).
In this problem, the goal is to simplify both chains simultaneously, so that the discrete \frechet\ distance between the resulting simplifications is bounded. More precisely,
given two chains $A$ and $B$ of lengths $m$ and $n$, respectively, an integer $k$ and three real numbers $\delta_1$,$\delta_2$,$\delta_3$, one needs to find two chains $A'$,$B'$ with vertices from $A$,$B$, respectively, each of length at most $k$, such that $d_1(A,A')\le\delta_1$, $d_2(B,B')\le\delta_2$, $\dfd(A',B')\le\delta_3$ ($d_1$ and $d_2$ can be any similarity measures and $\dfd$ is the discrete \frechet\ distance).
When the chains are simplified using the Hausdorff distance, i.e., $d_1,d_2$ is the Hausdorff distance (CPS-2H), the problem becomes  \npc~\cite{BeregJWYZ08}.
However, the complexity of the version in which $d_1,d_2$ is the discrete \frechet\ distance (CPS-3F) has been open since 2008.

\paragraph*{Related work.}

The \frechet\ distance and its variants have been studied extensively in the past two decades.
Alt and Godau~\cite{AltG95} gave an $O(mn\log mn)$-time algorithm for computing
the \frechet\ distance between two polygonal curves of lengths $m$ and $n$.
This result in the plane was recently improved by Buchin et al~\cite{BuchinBMM14}.
The discrete \frechet\ distance was originally defined by Eiter and Mannila~\cite{EiterM94}, who also
presented an $O(mn)$-time algorithm for computing it. A slightly sub-quadratic algorithm
was given recently by Agarwal et al.~\cite{AgarwalAKS14}.

As mentioned earlier, Bereg et al.~\cite{BeregJWYZ08} were the first to study simplification problems under the discrete \frechet\ distance.
They considered two such problems. In the first, the goal is to minimize the number of vertices
in the simplification, given a bound on the distance between the original chain and its simplification, and, in the second problem,
the goal is to minimize this distance, given a bound $k$ on the number of vertices in the simplification.
They presented an $O(n^2)$-time algorithm for the former problem and an $O(n^3)$-time algorithm for the latter problem, both using dynamic programming,
for the case where the vertices of the simplification are from the original chain.
(For the arbitrary vertices case, they solve the problems in
$O(n\log n)$ time and in $O(kn\log n\log(n/k))$ time, respectively.)
Driemel and Har-Peled~\cite{DriemelH13} showed how to preprocess a polygonal
curve in near-linear time and space, such that, given an integer $k > 0$,
one can compute a simplification in $O(k)$ time which has $2k-1$ vertices of the
original curve and is optimal up to a constant factor (w.r.t. the continuous \frechet\ distance),
compared to any curve consisting of $k$ arbitrary vertices.

For the chain pair simplification problem (CPS), Bereg et al.~\cite{BeregJWYZ08} proved that CPS-2H is \npc, and
conjectured that so is CPS-3F.
Wylie et al.~\cite{WylieLZ11} gave a heuristic algorithm for CPS-3F, using a greedy method with backtracking, and based on the assumption that
the (Euclidean) distance between adjacent $\alpha$-carbon atoms in a protein backbone is almost fixed.
More recently, Wylie and Zhu~\cite{WylieZ13} presented an approximation algorithm with approximation ratio 2 for the optimization version of CPS-3F.
Their algorithm actually solves the optimization version of a related problem called CPS-3$F^+$, it uses dynamic programming and its
running time is between $O(mn)$ and $O(m^2n^2)$ depending on the input simplification parameters.

Some special cases of CPS-3F have recently been studied.
Motivated by the need to reduce sensitivity to outliers when comparing curves, Ben Avraham et al.~\cite{AvrahamFKKS14}
studied the discrete \frechet\ distance with shortcuts problem.
In the one-sided variant, the dog is allowed to jump to any point that comes later in its sequence, rather than just to the
next point. The man has to visit the points in its sequence, one after the other, as in the standard discrete \frechet\ distance problem.
In the two-sided variant, both the man and the dog are allowed to skip points.
Unlike CPS-3F, the difference between an original chain and its simplification (in the two-sided variant) can be big, since the sole goal is to minimize the discrete \frechet\ distance between the two simplified chains.
(For this reason, Ben Avraham et al. do not allow both the man and the dog to move simultaneously, since, otherwise, they would both jump directly to their final points.) Moreover, the length of a simplification is only bounded by the length of the corresponding chain.
Both variants of the shortcuts problem can be solved in subquadratic time.

The one-sided variant of the (continuous) \frechet\ distance with shortcuts problem was studied by Driemel et al.~\cite{DriemelH13},
who considered the problem assuming the curves are \emph{c-packed} and shortcuts start and end at vertices of the noisy curve.
They gave a near-linear time $(3 + \varepsilon)$-approximation algorithm.
Buchin et al.~\cite{BuchinDS14} proved that the more general variant, where shortcuts can be taken at any point along the noisy curve, is \textbf{NP}-hard, and gave an $O(n^3 \log n)$-time 3-approximation algorithm for the corresponding decision problem.
Another approach for handling outliers (that is still somewhat related to our work) was proposed by Buchin et al.~\cite{BuchinBW09}, who studied the partial curve matching problem under the (continuous) \frechet\ distance. That is,
given two curves and a threshold $\delta$, find subcurves of maximum total length that are close to each other w.r.t. $\delta$.

\paragraph*{Our results.}

In Section~\ref{sec:weighted} we introduce the weighted chain pair simplification problem and prove that weighted CPS-3F is weakly \npc.
In Section~\ref{sec:cps3f}, we resolve the question concerning the complexity of CPS-3F by proving that it is polynomially solvable,
contrary to what was believed.
We do this by presenting a polynomial-time algorithm for the corresponding optimization problem.
We actually prove a stronger statement, implying, for example, that if weights are assigned to the vertices of only one of the chains, then the problem remains polynomially solvable.
In Section~\ref{sec:alg2} we devise a sophisticated $O(m^2n^2\min\{m,n\})$-time dynamic programming algorithm for the minimization problem of CPS-3F.
Besides being interesting from a theoretical point of view, only after developing (and implementing) this algorithm, were we able to apply the CPS-3F minimization problem to datasets from the Protein Data Bank (PDB), see below.

In section \ref{sec:simp} we study several less rigid variants of CPS-3F. In particular, we improve the result of Bereg et al.~\cite{BeregJWYZ08} mentioned above on the problem of finding the best simplification of a given length under the discrete \frechet\ distance, by presenting a more general $O(n^2 \log n)$-time algorithm (rather than an $O(n^3)$-time algorithm).

Finally, in Section~\ref{sec:empirical} we present some empirical results comparing (the minimization problems of) CPS-3$F^+$~\cite{WylieZ13} (the best available algorithm prior to this work) and CPS-3F using datasets from the PDB, and showing that with the latter we get much smaller simplifications (obeying the same distance bounds).

\section{Preliminaries} \label{sec:prelims}

Let $A=(a_1\ldots,a_m)$ and $B=(b_1,\ldots,b_n)$ be two sequences of $m$ and $n$ points, respectively, in $\reals^k$.
The discrete \frechet\ distance $\dfd(A,B)$ between $A$ and $B$ is defined as follows.
Fix a distance $\delta > 0$ and consider the Cartesian product $A \times B$ as the vertex set of
a directed graph $G_\delta$ whose edge set is
\begin{align*}
E_\delta = & \big\{ \big((a_i, b_j), (a_{i+1}, b_j)\big) \; | \; d(a_i, b_j), d(a_{i+1}, b_j) \le \delta \big\} \; \cup\\
           & \big\{ \big((a_i, b_j), (a_i, b_{j+1})\big) \; | \; d(a_i, b_j), d(a_i, b_{j+1}) \le \delta \big\} \; \cup\\
           & \big\{ \big((a_i, b_j), (a_{i+1}, b_{j+1})\big) \; | \; d(a_i, b_j), d(a_{i+1}, b_{j+1}) \le \delta \big\}\, .
\end{align*}
Then $\dfd(A,B)$ is the smallest $\delta > 0$ for which $(a_m, b_n)$ is reachable from $(a_1, b_1)$ in $G_\delta$.

\old{

Given two polygonal curves, the \frechet\ distance is usually defined as follows.

\begin{definition}
    The \frechet\ distance, $\df$, between two polygonal curves $f:[0,m] \rightarrow \reals^k$ and $g:[0,n] \rightarrow \reals^k$ is defined as:
    \[
    \df(f,g) = \min_{\substack{
                {\alpha:[0,1]\rightarrow[0,m]}\\
                {\beta:[0,1]\rightarrow[0,n]}}}\
                \max_{t \in [0,1]}\
                \Bigg\{d \Big ( f(\alpha(t)),\, g(\beta(t))\Big)\Bigg \}
    \]
    where $\alpha$ and $\beta$ range over all continuous non-decreasing functions with $\alpha(0)=0$, $\alpha(1)=m$, $\beta(0)=0$, $\beta(1)=n$.
\end{definition}

In the discrete case, we can adapt this definition as follows.
\begin{definition}
    The discrete \frechet\ distance, $\dfd$, between two polygonal curves $f:[0,m] \rightarrow \mathbb{R}^k$ and $g:[0,n] \rightarrow \mathbb{R}^k$ is defined as:
    \[
    \dfd(f,g) = \min_{\substack{
                    {\alpha:[1:m+n]\rightarrow[0:m]}\\
                    {\beta:[1:m+n]\rightarrow[0:n]}}}\
                    \max_{s \in [1:m+n]} \
                    \Bigg \{d \Big ( f(\alpha(s)), \, g(\beta(s)) \Big ) \Bigg \}
    \]
    where $\alpha$ and $\beta$ range over all discrete non-decreasing surjective mappings of the form $\alpha:[1:m+n]\rightarrow[0:m], \beta:[1:m+n]\rightarrow[0:n]$.
\end{definition}

The continuous \frechet\ distance is typically explained as the relationship
between a person and a dog connected by a leash walking along the two curves
and trying to keep the leash as short as possible. However, for the discrete case,
we only consider the nodes of these curves, and thus the man and dog must
``hop'' along the nodes.
}

\vspace{.5cm}
\noindent
The chain pair simplification problem (CPS) is formally defined as follows.
\begin{problem}[Chain Pair Simplification]\hfill \\
    \textbf{Instance:} Given a pair of polygonal chains $A$ and $B$ of lengths $m$ and $n$,
    respectively, an integer $k$, and three real numbers $\delta_1,\delta_2,\delta_3>0$.\\
    \textbf{Problem:} Does there exist a pair of chains $A'$,$B'$ each of at most $k$ vertices,
    such that the vertices of $A'$,$B'$ are from $A$,$B$, respectively, and
    $d_1(A,A') \le \delta_1$, $d_2(B,B') \le \delta_2$, and $\dfd(A',B') \le \delta_3$?
\end{problem}
When $d_1 = d_2 = d_H$, the problem is \npc\ and is called CPS-2H, and when $d_1 = d_2 = \dfd$, the problem is called CPS-3F.

\section{Weighted Chain Pair Simplification (WCPS-3F)} \label{sec:weighted}
We first introduce and consider a more general version of CPS-3F, namely, Weighted CPS-3F.
In the weighted version of the chain pair simplification problem, the vertices of the chains $A$ and $B$ are assigned arbitrary
weights, and, instead of limiting the length of the simplifications, one limits their weights. That is, the total weight of each simplification must not exceed a given value.
The problem is formally defined as follows.

\begin{problem}[Weighted Chain Pair Simplification] \hfill \\
    \textbf{Instance:} Given a pair of 3D chains $A$ and $B$, with lengths $m$ and $n$,
    respectively, an integer $k$, three real numbers $\delta_1,\delta_2,\delta_3>0$,
    and a weight function $C:\{a_1,\ldots,a_m,b_1,\ldots,b_n\} \rightarrow \reals^+$.\\
    \textbf{Problem:} Does there exist a pair of chains $A'$,$B'$ with $C(A'),C(B') \leq k$,
    such that the vertices of $A'$,$B'$ are from $A,B$ respectively, $d_1(A,A') \leq \delta_1$,
    $d_2(B,B') \leq \delta_2$, and $\dfd(A',B') \leq \delta_3$?
\end{problem}
When $d_1 = d_2 = \dfd$, the problem is called WCPS-3F.
When $d_1 = d_2 = d_H$, the problem is \npc, since the non-weighted version (i.e., CPS-2H) is already \npc~\cite{BeregJWYZ08}.

We prove that WCPS-3F is weakly \npc\ via a reduction from the \emph{set partition} problem:
Given a set of positive integers $S=\{s_1,\dots,s_n\}$,
find two sets $P_1,P_2 \subset S$ such that $P_1 \cap P_2 = \emptyset$,
$P_1 \cup P_2 = S$, and the sum of the numbers in $P_1$ equals the sum of the numbers in $P_2$.
This is a weakly \npc\ special case of the classic subset-sum problem.

Our reduction builds two curves with weights reflecting the values in $S$. We think of the two
curves as the subsets of the partition of $S$. Although our problem requires positive weights,
we also allow zero weights in our reduction for clarity. Later, we show how to remove these weights by
slightly modifying the construction.

\begin{figure}[htb]
    \centering \includegraphics[scale=0.25]{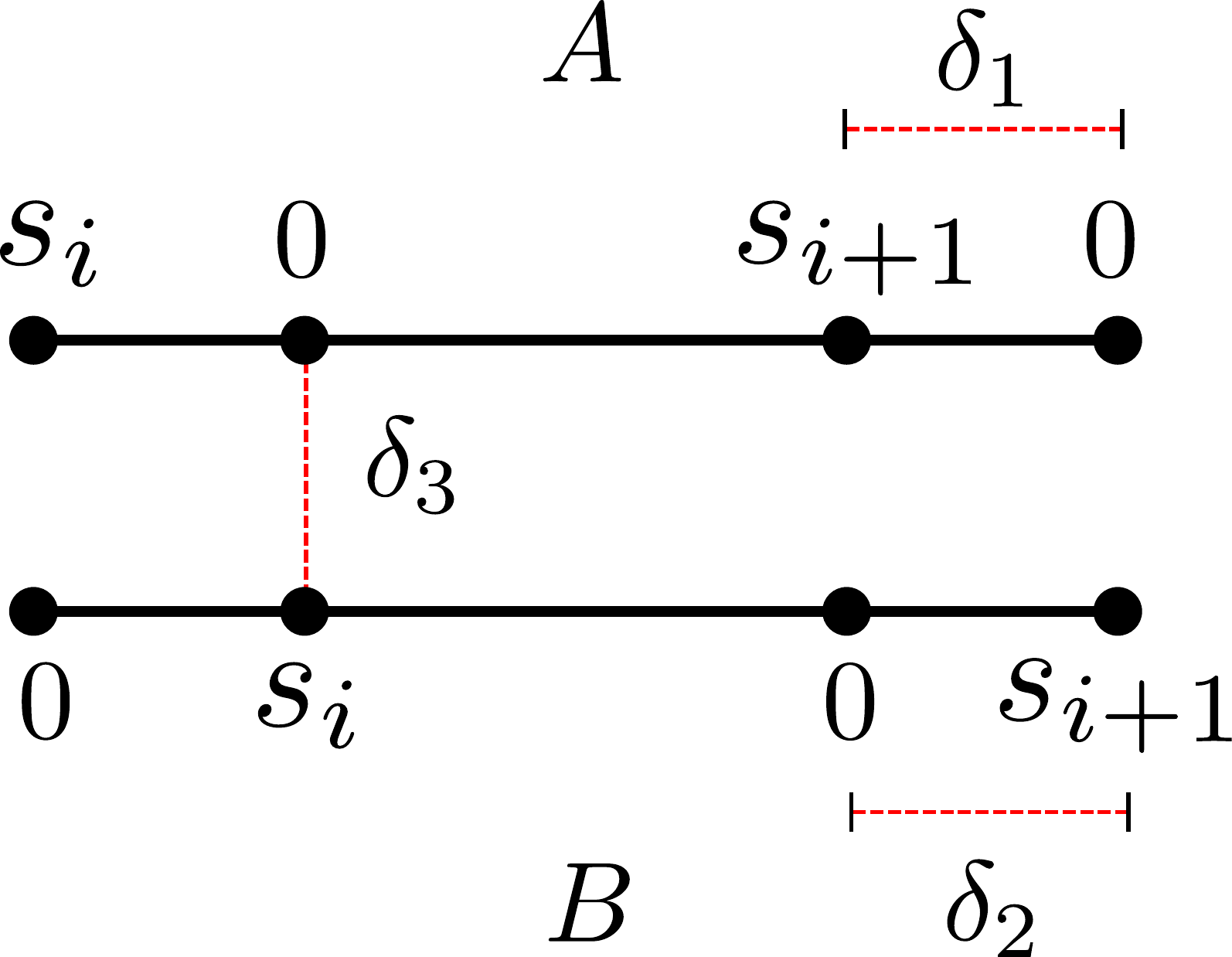}
    \centering \caption{The reduction for the weighted chain pair simplification problem under the discrete \frechet\ distance.}
    \label{fig:reduction}
\end{figure}

\begin{theorem}
    The weighted chain pair simplification problem under the discrete \frechet\ distance is weakly \npc.
\end{theorem}

\begin{proof}
Given the set of positive integers $S=\{s_1,\dots, s_n\}$, we construct two curves $A$ and $B$ in the plane, each of length $2n$.
We denote the weight of a vertex $x_i$ by $w(x_i)$.
$A$ is constructed as follows.
The $i$'th odd vertex of $A$ has weight $s_i$, i.e. $w(a_{2i-1})=s_i$, and coordinates $a_{2i-1}=(i,1)$.
The $i$'th even vertex of $A$ has coordinates $a_{2i}=(i+0.2,1)$ and weight zero.
Similarly, the $i$'th odd vertex of $B$ has weight zero and coordinates $b_{2i-1}=(i,0)$, and
the $i$'th even vertex of $B$ has coordinates $b_{2i}=(i+0.2,0)$ and weight $s_i$, i.e. $w(b_{2i})=s_i$.
Figure~\ref{fig:reduction} depicts the vertices $a_{2i-1},a_{2i},a_{2(i+1)-1},a_{2(i+1)}$ of $A$ and $b_{2i-1},b_{2i},b_{2(i+1)-1},b_{2(i+1)}$ of $B$.
Finally, we set $\delta_1=\delta_2=0.2$, $\delta_3=1$, and $k = \mathfrak{S}$,
where $\mathfrak{S}$ denotes the sum of the elements of $S$ (i.e., $\mathfrak{S}=\sum_{j=1}^n s_j$).

We claim that $S$ can be partitioned into two subsets, each of sum $\mathfrak{S}/2$, if and only if $A$ and $B$ can be simplified with
the constraints $\delta_1=\delta_2=0.2$, $\delta_3=1$ and $k=\mathfrak{S}/2$, i.e.,
$C(A'),C(B') \leq \mathfrak{S}/2$.

First, assume that $S$ can be partitioned into sets $S_A$ and $S_B$, such that $\sum_{s \in S_A} s = \sum_{s \in S_B} s = \mathfrak{S}/2$.
We construct simplifications of $A$ and of $B$ as follows.
\[
A' = \{a_{2i-1} \ | \ s_i \in S_A \} \cup \{a_{2i} \ | \ s_i \notin S_A \} \mbox{ and }
B' = \{b_{2i} \ | \ s_i \in S_B \} \cup \{b_{2i-1} \ | \ s_i \notin S_B \}\, .
\]
It is easy to see that $C(A'),C(B') \leq \mathfrak{S}/2$.
Also, since $\{S_A,S_B\}$ is a partition of $S$,
exactly one of the following holds, for any $1 \le i \le n$:
 \begin{enumerate}
 \item $a_{2i-1}\in A', b_{2i-1}\in B'$ \mbox{ and } $a_{2i}\notin A', b_{2i}\notin B'$.
 \item $a_{2i-1}\notin A', b_{2i-1}\notin B'$ \mbox{ and} $a_{2i}\in A', b_{2i}\in B'$.
 \end{enumerate}
This implies that $\dfd(A,A')\le 0.2=\delta_1$, $\dfd(B,B')\le 0.2=\delta_2$ and $\dfd(A',B')\le 1=\delta_3$.

Now, assume there exist simplifications $A', B'$ of $A, B$, such that
$\dfd(A,A') \le \delta_1=0.2$, $\dfd(B,B') \le \delta_2=0.2$, $\dfd(A',B') \le \delta_3=1$, and $C(A'), C(B') \leq k=\mathfrak{S}/2$.
Since $\delta_1=\delta_2=0.2$, for any $1 \le i \le n$, the simplification $A'$ must contain one of $a_{2i-1},a_{2i}$,
and the simplification $B'$ must contain one of $b_{2i-1},b_{2i}$.
Since $\delta_3=1$, for any $i$, at least one of the following two conditions holds: $a_{2i-1} \in A'$ and $b_{2i-1} \in B'$  or
$a_{2i} \in A'$ and $b_{2i} \in B'$.
Therefore, for any $i$, either $a_{2i-1} \in A$ or $b_{2i} \in B$, implying that $s_i$ participates in either $C(A')$ or $C(B')$.
However, since $C(A'), C(B') \le \mathfrak{S}/2$, $s_i$ cannot participate in both $C(A')$ and $C(B')$.
It follows that $C(A') = C(B') = \mathfrak{S}/2$, and we get a partition of $S$ into two sets, each of sum $\mathfrak{S}/2$.

Finally, we note that WCPS-3F is in \textbf{NP}.  For an instance $I$ with chains $A,B$, given simplifications $A',B'$, we can
verify in polynomial time that $\dfd(A,A') \le \delta_1$, $\dfd(B,B') \le \delta_2$, $\dfd(A',B') \le \delta_3$, and $C(A'),C(B') \le k$.
\end{proof}

Although our construction of $A'$ and $B'$ uses zero weights, a simple modification enables us to prove that the problem is weakly \npc\
also when only positive integral weights are allowed. Increase all the weights by 1, that is,
$w(a_{2i-1})=w(b_{2i})=s_i+1$ and $w(a_{2i})=w(b_{2i-1})=1$, for $1 \le i \le n$, and set $k = \mathfrak{S}/2 + n$.
It is easy to verify that our reduction still works.
Finally, notice that we could overlay the two curves choosing $\delta_3=0$ and prove that the problem
is still weakly \npc\ in one dimension.

\section{Chain Pair Simplification (CPS-3F)} \label{sec:cps3f}

We now turn our attention to CPS-3F, which is the special case of WCPS-3F where each vertex has weight one.

We present an algorithm for the minimization version of CPS-3F. That is, we compute the minimum integer $k^*$, such that
there exists a ``walk'', as above, in which each of the dogs makes at most $k^*$ hops.
The answer to the decision problem is ``yes'' if and only if $k^* < k$.

Returning to the analogy of the man and the dog, we can extend it as follows.
Consider a man and his dog connected by a leash of length $\delta_{1}$,
and a woman and her dog connected by a leash of length $\delta_{2}$.
The two dogs are also connected to each other by a leash of length $\delta_{3}$.
The man and his dog are walking on the points of a chain $A$ and the woman and
her dog are walking on the points of a chain $B$. The dogs may skip points.
The problem is to determine whether there exists a ``walk'' of the man and
his dog on $A$ and the woman and her dog on $B$, such that each of the dogs steps on at most $k$ points.

\paragraph*{Overview of the algorithm.}
We say that $(a_{i},a_{p},b_{j},b_{q})$ is a \emph{possible} configuration of the man, woman and the two dogs on the paths $A$ and $B$,
if $d(a_{i},a_{p}) \le \delta_{1}$, $d(b_{j},b_{q}) \le \delta_{2}$ and $d(a_{p},b_{q}) \le \delta_{3}$.
Notice that there are at most $m^{2}n^{2}$ such configurations.
Now, let $G$ be the DAG whose vertices are the possible configurations, such that
there exists a (directed) edge from vertex $u=(a_{i},a_{p},b_{j},b_{q})$ to vertex $v=(a_{i'},a_{p'},b_{j'},b_{q'})$ if and only if
our gang can move from configuration $u$ to configuration $v$.
That is, if and only if $i \le i' \le i+1$, $p \le p'$, $j \le j' \le j+1$, and $q \le q'$.
Notice that there are no cycles in $G$ because backtracking is forbidden.
For simplicity, we assume that the first and last points of $A'$ (resp., of $B'$) are $a_1$ and $a_m$ (resp., $b_1$ and $b_n$),
so the initial and final configurations are $s=(a_1,a_1,b_1,b_1)$ and $t=(a_m,a_m,b_n,b_n)$, respectively.
(It is easy, however, to adapt the algorithm below to the case where the initial and final points of $A'$ and $B'$ are not specified,
see remark below.)
Our goal is to find a path from $s$ to $t$ in $G$. However, we want each of our dogs to
step on at most $k$ points, so, instead of searching for any path from $s$ to $t$, we search for
a path that minimizes the value $max\{|A'|,|B'|\}$,
and then check if this value is at most $k$.

For each edge $e=(u,v)$, we assign two weights, $w_{A}(e),w_{B}(e)\in\{0,1\}$,
in order to compute the number of hops in $A'$ and in $B'$, respectively.
$w_{A}(u,v)=1$ if and only if the first dog jumps to a new point between configurations
$u$ and $v$ (i.e., $p < p'$), and, similarly,
$w_{B}(u,v)=1$ if and only if the second dog jumps to a new point between $u$ and $v$ (i.e., $q < q'$).
Thus, our goal is to find a path $P$ from $s$ to $t$ in $G$, such that $max\{\underset{e\in P} \sum w_{A}(e),\underset{e\in P}\sum w_{B}(e)\}$
is minimized.

Assume w.l.o.g. that $m \le n$.
Since $|A'| \le m$ and $|B'| \le n$, we maintain, for each vertex $v$ of $G$, an array $X(v)$ of size $m$,
where $X(v)[r]$ is the minimum number $z$ such that $v$ can be reached from $s$
with (at most) $r$ hops of the first dog and $z$ hops of the second dog.
We can construct these arrays by processing the vertices of $G$ in topological order
(i.e., a vertex is processed only after all its predecessors have been processed).
This yields an algorithm of running time $O(m^{3}n^{3}\min\{m,n\})$, as described in Algorithm~\ref{alg:cps3f}.

\begin{algorithm}[!ht]
\small
    \vspace*{.25cm}
    \begin{enumerate}
        \item \label{step:graph}Create a directed graph $G=(V,E)$ with two weight functions $w_{A},\, w_{B}$, such that:
        \begin{itemize}
            \item $V$ is the set of all configurations $(a_{i},a_{p},b_{j},b_{q})$ with $d(a_{i},a_{p})\le\delta_{1}$, $d(b_{j},b_{q})\le\delta_{2}$, and $d(a_{p},b_{q})\le\delta_{3}$.
            \item $E=\{((a_{i},a_{p},b_{j},b_{q}),(a_{i'},a_{p'},b_{j'},b_{q'}))\,|\, i\le i'\le i+1,\, p\le p',\, j\le j'\le j+1,\, q\le q'\}$.
            \item For each $((a_{i},a_{p},b_{j},b_{q}),(a_{i'},a_{p'},b_{j'},b_{q'})) \in E$, set
            \begin{itemize}
                \item $w_{A}((a_{i},a_{p},b_{j},b_{q}),(a_{i'},a_{p'},b_{j'},b_{q'}))=
                   \begin{cases}
                         1, & p<p'\\
                         0, & otherwise
                   \end{cases}$
                \item $w_{B}((a_{i},a_{p},b_{j},b_{q}),(a_{i'},a_{p'},b_{j'},b_{q'}))=
                   \begin{cases}
                         1, & q<q'\\
                         0, & otherwise
                   \end{cases}$
            \end{itemize}
        \end{itemize}
        \item \label{step:sort}Sort $V$ topologically.
        \item \label{step:init}Initialize the array $X(s)$ (i.e., set $X(s)[r] = 0$, for $r = 0,\ldots,m-1$).
        \item \label{step:calc}For each $v \in V \setminus \{s\}$ (advancing from left to right in the sorted sequence) do:
        \begin{enumerate}
            \item Initialize the array $X(v)$ (i.e., set $X(v)[r] = \infty$, for $r = 0,\ldots,m-1$).
            \item For each $r$ between $0$ and $m-1$, compute $X(v)[r]$:
        \end{enumerate}
            $\hspace{1cm} X(v)[r]=
            \underset{\begin{array}{c}(u,v)\in E\end{array}}{\min}
            \begin{cases}
                X(u)[r]+w_{B}(u,v), & w_{A}(u,v)=0\\
                X(u)[r-1]+w_{B}(u,v), & w_{A}(u,v)=1
            \end{cases}$
        \item \label{step:return}Return $k^* = \underset{r}{\min}\max\{r,\, X(t)[r]\}$\,.
    \end{enumerate}
    \caption{CPS-3F}
    \label{alg:cps3f}
\end{algorithm}

\paragraph*{Running time.}
The number of vertices in $G$ is $|V|=O(m^2 n^2)$.
By the construction of the graph, for any vertex $(a_{i},a_{p},b_{j},b_{q})$ the maximum number of outgoing edges is $O(mn)$.
So we have $|E|=O(|V|mn)=O(m^3 n^3)$. Thus, constructing the graph $G$ in Step~\ref{step:graph} takes $O(n^3m^3)$ time.
Step~\ref{step:sort} takes $O(|E|)$ time, while Step~\ref{step:init} takes $O(m)$ time.
In Step~\ref{step:calc}, for each vertex $v$ and for each index $r$, we consider all configurations that can directly precede $v$.
So each edge of $G$ participates in exactly $m$ minimum computations, implying that Step~\ref{step:calc} takes $O(|E|m)$ time.
Step~\ref{step:return} takes $O(m)$ time.
Thus, the total running time of the algorithm is $O(m^4 n^3)$.

\begin{theorem}
The chain pair simplification problem under the discrete \frechet\ distance (CPS-3F) is polynomial, i.e., CPS-3F $\in$ \textbf{P}.
\end{theorem}

\begin{remark}
\label{rem:init-final-conf}
As mentioned, we have assumed that the first
and last points of $A'$ (resp., $B'$) are $a_1$ and $a_m$ (resp., $b_1$ and $b_n$), so we have
a single initial configuration (i.e., $s=(a_1,a_1,b_1,b_1)$) and
a single final configuration (i.e., $t=(a_m,a_m,b_n,b_n)$).
However, it is easy to adapt our algorithm to the case where the first and last points of the chains $A'$ and $B'$ are not specified.
In this case, any possible configuration of the form $(a_1,a_p,b_1,b_q)$ is considered a potential initial configuration,
and any possible configuration of the form
$(a_m,a_p,b_n,b_q)$ is considered a potential final configuration, where $1\le p\le m$ and $1\le q\le n$.
Let $S$ and $T$ be the sets of potential initial and final configurations, respectively.
(Then, $|S|=O(mn)$ and $|T|=O(mn)$.)
We thus remove from $G$ all edges entering a potential initial configuration, so that each such configuration becomes a ``root'' in the (topologically) sorted sequence.
Now, in Step~\ref{step:init} we initialize the arrays of each $s\in S$ in total time $O(m^2n)$, and in Step~\ref{step:calc} we only process the vertices that are not in $S$.
The value $X(v)[r]$ for such a vertex $v$ is now the minimum number $z$ such that $v$ can be reached from $s$ with $r$ hops of the first dog and $z$ hops of the second dog, over $\emph{all}$ potential initial configurations $s \in S$.
In the final step of the algorithm, we calculate the value $k^{*}$ in $O(m)$ time, for each potential final configuration $t\in T$.
The smallest value obtained is then the desired value.
Since the number of potential final configurations is only $O(mn)$, the total running time of the final step of the algorithm is only $O(m^2n)$, and the running time of the entire algorithm remains $O(m^4n^3)$.
\end{remark}

\subsection{The weighted version}
Weighted CPS-3F, which was shown to be weakly \npc\ in the previous section, can be solved in a similar manner, albeit with running time that depends on the number of different point weights in chain $A$ (alternatively, $B$).
We now explain how to adapt our algorithm to the weighted case.
We first redefine the weight functions $w_{A}$ and $w_{B}$ (where $C(x)$ is the weight of point $x$):
\begin{itemize}
		\item $w_{A}((a_{i},a_{p},b_{j},b_{q}),(a_{i'},a_{p'},b_{j'},b_{q'}))=
		          \begin{cases}
                C(a_{p'}), & p<p'\\
                0, & otherwise
              \end{cases}$
    \item $w_{B}=((a_{i},a_{p},b_{j},b_{q}),(a_{i'},a_{p'},b_{j'},b_{q'}))
              \begin{cases}
                C(b_{q'}), & q<q'\\
                0, & otherwise
              \end{cases}$
\end{itemize}
Next, we increase the size of the arrays $X(v)$ from $m$ to the number of different weights that can be obtained by a subset of $A$ (alternatively, $B$).
(For example, if $|A|=3$ and $C(a_1) = 2$, $C(a_2) = 2$, and $C(a_3) = 4$, then the weights that can be obtained are $2,4,2+4=6,2+2+4=8$, so the size of the arrays would be 4.) Let $c[r]$ be the $r$'th largest such weight. Then $X(v)[r]$ is the minimum number $z$, such that $v$ can be reached from $s$ with hops of total weight (at most) $c[r]$ of the first dog and hops of total weight $z$ of the second dog. $X(v)[r]$ is calculated as follows:
\[
            \hspace{1cm}
            X(v)[r]=
            \underset{\begin{array}{c}(u,v)\in E\end{array}}{\min}
            \begin{cases}
                X(u)[r]+w_{B}(u,v), & w_{A}(u,v)=0\\
                X(u)[r']+w_{B}(u,v), & w_{A}(u,v)>0
            \end{cases}
            \ ,
\]
where $c[r']=c[r]-w_{A}(u,v)$.
If the number of different weights that can be obtained by a subset of $A$ (alternatively, $B$) is $f(A)$ (resp., $f(B)$), then the running time is $O(m^{3}n^{3}f(A))$ (resp., $O(m^{3}n^{3}f(B))$), since the only change that affects the running time is the size of the arrays $X(v)$.
We thus have
\begin{theorem}
The weighted chain pair simplification problem under the discrete \frechet\ distance (Weighted CPS-3F) (and its corresponding minimization problem) can be solved in $O(m^{3}n^{3}\min\{f(A),f(B)\})$ time, where $f(A)$ (resp., $f(B)$) is the number of different weights that can be obtained by a subset of $A$ (resp., $B$). In particular, if only one of the chains, say $B$, has points with non-unit weight, then $f(A)=O(m)$, and the running time is polynomial; more precisely, it is $O(m^{4}n^{3})$.
\end{theorem}

\begin{remark}
We presented an algorithm that minimizes $\max\{|A'|,|B'|\}$ given the error parameters $\delta_1, \delta_2, \delta_3$.
Another optimization version of CPS-3F is to minimize, e.g., $\delta_3$ (while obeying the requirements specified by $\delta_1, \delta_2$ and $k$).
It is easy to see that Algorithm~\ref{alg:cps3f} can be adapted to solve this version within roughly the same time bound.
\end{remark}

\section{An Efficient Implementation} \label{sec:alg2}

The time and space complexity of Algorithm~\ref{alg:cps3f} (which is $O(m^{3}n^{3}\min\left\{ m,n\right\})$ and $O(m^3n^3)$, respectively)
makes it impractical for our motivating biological application (as $m, n$ could be 500$\sim$600); see Section~\ref{sec:empirical}.
In this section, we show how to reduce the time and space bounds by a factor of $mn$, using dynamic programming.

We generate all configurations of the form $(a_{i},a_{p},b_{j},b_{q})$, where the outermost for-loop is governed by $i$, the next level loop by $j$, then $p$, and finally $q$.
When a new configuration $v=(a_{i},a_{p},b_{j},b_{q})$ is generated, we first check whether it is \emph{possible}. If it is not possible, we set $X(v)[r]=\infty$, for $1 \le r \le m$, and if it is, we compute $X(v)[r]$, for $1 \le r \le m$.

We also maintain for each pair of indices $i$ and $j$, three tables $C_{i,j}$, $R_{i,j}$, $T_{i,j}$ that assist us
in the computation of the values $X(v)[r]$:
\begin{alignat*}{1}
C_{i,j}[p,q,r] & =\min_{\substack{{1\le p'\le p}}}X(a_{i},a_{p'},b_{j},b_{q})[r]\\
R_{i,j}[p,q,r] & =\min_{\substack{{1\le q'\le q}}}X(a_{i},a_{p},b_{j},b_{q'})[r]\\
T_{i,j}[p,q,r] & =\min_{\substack{{1\le p'\le p}\\{1\le q'\le q}}}X(a_{i},a_{p'},b_{j},b_{q'})[r]
\end{alignat*}
Notice that the value of cell $[p,q,r]$ is determined by the value of one or two previously-determined cells and
$X(a_{i},a_{p},b_{j},b_{q})[r]$ as follows:
\begin{alignat*}{1}
C_{i,j}[p,q,r] & =\min\{C_{i,j}[p-1,q,r],X(a_{i},a_{p},b_{j},b_{q})[r]\}\\
R_{i,j}[p,q,r] & =\min\{R_{i,j}[p,q-1,r],X(a_{i},a_{p},b_{j},b_{q})[r]\}\\
T_{i,j}[p,q,r] & =\min\{T_{i,j}[p-1,q,r],T_{i,j}[p,q-1,r],X(a_{i},a_{p},b_{j},b_{q})[r]\}
\end{alignat*}

Observe that in any configuration
that can immediately precede the current configuration $(a_{i},a_{p},b_{j},b_{q})$, the man is either at $a_{i-1}$ or at $a_i$ and the woman is either at $b_{j-1}$ or at $b_j$ (and the dogs are at $a_{p'}$, $p' \le p$, and $b_{q'}, q' \le q$, respectively).
The ``saving'' is achieved, since now we only need to access a constant number of table entries in order to compute
the value $X(a_{i},a_{p},b_{j},b_{q})[r]$.

\old{
The running time of the above algorithm for CPS-3F is extremely high.
We now show how to reduce the running time using a dynamic programming algorithm.

In order to do this, we generate \emph{all} the configurations of the form $(a_{i},a_{p},b_{j},b_{q})$, where the first level for-loop is governed by $i$, the second level loop by $j$, the third by $p$, and the fourth by $q$.
When a new configuration $v=(a_{i},a_{p},b_{j},b_{q})$ is generated, we first check whether it is valid. If it is not valid, we set $X(v)[r]=\infty$, for $1\le r\le m$, and, if it is valid, we compute $X(v)[r]$, for $1\le r\le m$.

We also maintain for each pair of indices $i$ and $j$, three tables $C_{i,j}$, $R_{i,j}$, $T_{i,j}$, that assist us in the computation of the values $X(v)[r]$:
\begin{alignat*}{1}
C_{i,j}[p,q,r] & =\min_{\substack{{1\le p'\le p}}}X(a_{i},a_{p'},b_{j},b_{q})[r]\\
R_{i,j}[p,q,r] & =\min_{\substack{{1\le q'\le q}}}X(a_{i},a_{p},b_{j},b_{q'})[r]\\
T_{i,j}[p,q,r] & =\min_{\substack{{1\le p'\le p}\\{1\le q'\le q}}}X(a_{i},a_{p'},b_{j},b_{q'})[r].
\end{alignat*}

Notice that the value of cell $[p,q,r]$ is determined by the value of one or two previous cells and $X(a_{i},a_{p},b_{j},b_{q})[r]$ as follows:
\begin{alignat*}{1}
C_{i,j}[p,q,r] & =\min\{C_{i,j}[p-1,q,r],X(a_{i},a_{p},b_{j},b_{q})[r]\}\\
R_{i,j}[p,q,r] & =\min\{R_{i,j}[p,q-1,r],X(a_{i},a_{p},b_{j},b_{q})[r]\}\\
T_{i,j}[p,q,r] & =\min\{T_{i,j}[p-1,q,r],T_{i,j}[p,q-1,r],X(a_{i},a_{p},b_{j},b_{q})[r]\}.
\end{alignat*}

Observe that in any configuration that can immediately
precede the current configuration
$(a_{i},a_{p},b_{j},b_{q})$, the man is either
at $a_{i-1}$ or at $a_i$ and the woman is either
at $b_{j-1}$ or at $b_j$ (and the dogs are at $a_{k'}$, $k'\le k$, and $b_{l'}$, $l'\le l$, respectively). The “saving” is achieved, since now we only need to access a constant number of table entries in order to compute the value $X(a_{i},a_{p},b_{j},b_{q})[r]$.
}

One can illustrate the algorithm using the matrix in Figure \ref{fig:matrix}.
There are $mn$ large cells, each of them containing a matrix of size
$mn$. The large cells correspond to the positions of the man and
the woman. The inner matrices correspond to the positions of the
two dogs (for given positions of the man and woman). Consider an
optimal ``walk'' of the gang that ends at cell $(a_{i},a_{p},b_{j},b_{q})$
(marked by a full circle), such that the first dog has visited $r$ points.
The previous cell in this ``walk'' must be
in one of the 4 large cells $(a_{i},b_{j})$,$(a_{i-1},b_{j})$,$(a_{i},b_{j-1})$,$(a_{i-1},b_{j-1})$.
Assume, for example, that it is in $(a_{i-1},b_{j})$.
Then, if it is in the blue area,
then $X(a_{i},a_{p},b_{j},b_{q})[r]=C_{i-1,j}[p-1,q,r-1]$ (marked by an empty square), since
only the position of the first dog has changed when the gang moved to $(a_{i},a_{p},b_{j},b_{q})$.
If it is in the purple area, then
$X(a_{i},a_{p},b_{j},b_{q})[r]=R_{i-1,j}[p,q-1,r]+1$ (marked by a x),
since only the position of the second dog has changed. If it is in
the orange area, then $X(a_{i},a_{p},b_{j},b_{q})[r]=T_{i-1,j}[p-1,q-1,r-1]+1$ (marked by an empty circle),
since the positions of both dogs have changed. Finally, if it is the cell marked by the full square,
then simply $X(a_{i},a_{p},b_{j},b_{q})[r]=X(a_{i-1},a_{p},b_{j},b_{q})[r]$, since
both dogs have not moved.
The other three cases, in which the previous cell is in one of the 3 large cells
$(a_{i},b_{j})$,$(a_{i},b_{j-1})$,$(a_{i-1},b_{j-1})$, are handled similarly.

\begin{figure}[thb]
  \begin{center}
    \includegraphics[scale=1.2]{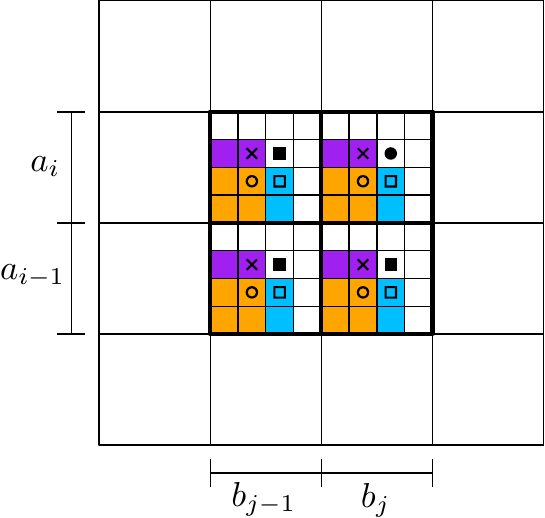}
  \end{center}
  \caption{Illustration of Algorithm~\ref{alg:cps3fdp}.}
  \label{fig:matrix}
\end{figure}

We are ready to present the dynamic programming algorithm.
The initial configurations correspond to cells in the large cell $(a_1,b_1)$.
For each initial configuration $(a_{1},a_{p},b_{1},b_{q})$, we set $X(a_{1},a_{p},b_{1},b_{q})[1]=1$.

\begin{algorithm}[H]
\begin{description}
\item [{for}] $i=1$ to $m$
\begin{description}
\item [{for}] $j=1$ to $n$
\begin{description}
\item [{for}] $p=1$ to $m$
\begin{description}
\item [{for}] $q=1$ to $n$
\begin{description}
\item [{for}] $r=1$ to $m$

$\begin{aligned}
&X_{(-1,0)}&=&\min\begin{cases}
C_{i-1,j}[p-1,q,r-1]\\
R_{i-1,j}[p,q-1,r]+1\\
T_{i-1,j}[p-1,q-1,r-1]+1\\
X(a_{i-1},a_{p},b_{j},b_{q})[r]
\end{cases}\\
&X_{(0,-1)}&=&\min\begin{cases}
C_{i,j-1}[p-1,q,r-1]\\
R_{i,j-1}[p,q-1,r]+1\\
T_{i,j-1}[p-1,q-1,r-1]+1\\
X(a_{i},a_{p},b_{j-1},b_{q})[r]
\end{cases}\\
&X_{(-1,-1)}&=&\min\begin{cases}
C_{i-1,j-1}[p-1,q,r-1]\\
R_{i-1,j-1}[p,q-1,r]+1\\
T_{i-1,j-1}[p-1,q-1,r-1]+1\\
X(a_{i-1},a_{p},b_{j-1},b_{q})[r]
\end{cases}\\
&X_{(0,0)}&=&\min\begin{cases}
C_{i,j}[p-1,q,r-1]\\
R_{i,j}[p,q-1,r]+1\\
T_{i,j}[p-1,q-1,r-1]+1
\end{cases}\\
\end{aligned}$

$X(a_{i},a_{p},b_{j},b_{q})[r]=\min\{X_{(-1,0)},X_{(0,-1)},X_{(-1,-1)},X_{(0,0)}\}$\\

$\begin{aligned}
&C_{i,j}[p,q,r] & =&\min\{C_{i,j}[p-1,q,r],X(a_{i},a_{p},b_{j},b_{q})[r]\}\\
&R_{i,j}[p,q,r] & =&\min\{R_{i,j}[p,q-1,r],X(a_{i},a_{p},b_{j},b_{q})[r]\}\\
&T_{i,j}[p,q,r] & =&\min\{T_{i,j}[p-1,q,r],T_{i,j}[p,q-1,r],X(a_{i},a_{p},b_{j},b_{q})[r]\}
\end{aligned}$

\end{description}
\end{description}
\end{description}
\end{description}
\item [{return}] $\underset{r,p,q}{\min}\max\{r,X(a_{m},a_{p},b_{n},b_{q})[r]\}$
\end{description}
\caption{CPS-3F using dynamic programming}
\label{alg:cps3fdp}
\end{algorithm}

\begin{theorem}
The minimization version of the chain pair simplification problem under the discrete \frechet\ distance (CPS-3F) can be solved in $O(m^{2}n^{2}\min\left\{ m,n\right\})$ time.
\end{theorem}

\section{1-Sided Chain Pair Simplification} \label{sec:simp}

Sometimes, one of the two input chains, say $B$, is much shorter than the other, possibly because it has already been simplified.
In these cases, we only want to simplify $A$, in a way that maintains the resemblance between the two input chains.
We thus define the 1-sided chain pair simplification problem.

\begin{problem}[1-Sided Chain Pair Simplification]\hfill \\
    \textbf{Instance:} Given a pair of polygonal chains $A$ and $B$ of lengths $m$ and $n$,
    respectively, an integer $k$, and two real numbers $\delta_1,\delta_3>0$.\\
    \textbf{Problem:} Does there exist a chain $A'$ of at most $k$ vertices,
    such that the vertices of $A'$ are from $A$, $\dfd(A,A') \le \delta_1$, and $\dfd(A',B) \le \delta_3$?
\end{problem}

The optimization version of this problem can be solved using similar ideas to those used in the solution of the 2-sided problem.
Here a \emph{possible} configuration is a 3-tuple $(a_{i},a_{p},b_{j})$, where $d(a_i,a_p) \le \delta_1$ and $d(a_p,b_j) \le \delta_3$.
We construct a graph and find a shortest path from one of the starting configurations to one of the final configurations; see Algorithm~\ref{alg:cps3f1s}. Arguing as for Algorithm~\ref{alg:cps3f}, we get that $|V|=O(m^2 n)$ and $|E|=O(|V|m)=O(m^3 n)$. Moreover, it is easy to see that
the running time of Algorithm~\ref{alg:cps3f1s} is $O(m^3 n)$, since it does not maintain an array for each vertex.

\begin{algorithm}[!ht]
\small
    \vspace*{.25cm}
    \begin{enumerate}
        \item Create a directed graph $G=(V,E)$ with a weight function $w$, such that:
        \begin{itemize}
            \item $V=\{(a_{i},a_{p},b_{j})\ |\ d(a_{i},a_{p})\le\delta_{1} \ \mbox{and} \ d(a_{p},b_{j})\le\delta_{3}\}$.
            \item $E=\{((a_{i},a_{p},b_{j}),(a_{i'},a_{p'},b_{j'}))\ |\ i\le i'\le i+1,\, p\le p',\, j\le j'\le j+1\}$.
            \item For each $((a_{i},a_{p},b_{j}),(a_{i'},a_{p'},b_{j'})) \in E$, set
                    \[
                        w((a_{i},a_{p},b_{j}),(a_{i'},a_{p'},b_{j'})=
                       \begin{cases}
                             1, & p<p'\\
                             0, & otherwise
                       \end{cases}
                    \]
            \item Let $S$ be the set of starting configurations and let $T$ be the set of final configurations.
        \end{itemize}
        \item Sort $V$ topologically.
        \item Set $X(s)=0$, for each $s \in S$.
        \item For each $v \in V \setminus S$ (advancing from left to right in the sorted sequence) do:
            \[
            X(v)=\underset{(u,v)\in E}{\min}\ \{X(u)+w(u,v)\}.
            \]
        \item Return $k^* = \underset{t\in T}{\min}\ X(t)$.
    \end{enumerate}
    \caption{1-sided CPS-3F}
    \label{alg:cps3f1s}
\end{algorithm}

To reduce the running time we use dynamic programming as in Section~\ref{sec:alg2}.
We generate all configurations of the form $(a_{i},a_{p},b_{j})$.
When a new configuration $v=(a_{i},a_{p},b_{j})$ is generated, we first check whether it is \emph{possible}. If it is not possible, we set $X(v)=\infty$, and if it is, we compute $X(v)$.

We also maintain for each pair of indices $i$ and $j$, a table $A_{i,j}$ that assists us
in the computation of the value $X(v)$:
$$
A_{i,j}[p]=\min_{\substack{{1\le p'\le p}}} X(a_{i},a_{p'},b_{j})\, .
$$
Notice that $A_{i,j}[p]$ is the minimum of $A_{i,j}[p-1]$ and $X(a_{i},a_{p},b_{j})$.

We observe once again that in any configuration
that can immediately precede the current configuration $(a_{i},a_{p},b_{j})$, the man is either at $a_{i-1}$ or at $a_i$ and the woman is either at $b_{j-1}$ or at $b_j$ (and the dog is at $a_{p'}$, $p' \le p$).
The ``saving'' is achieved, since now we only need to access a constant number of table entries in order to compute
the value $X(a_{i},a_{p},b_{j})$.
We obtain the following dynamic programming algorithm whose running time is $O(m^2n)$.

\old{
\begin{figure*}[htp]
  \begin{center}
    \includegraphics{matrix1sided.pdf}
  \end{center}
  \caption{The vertex matrix.}
  \label{fig:matrix1sided}
\end{figure*}

We can illustrate the algorithm using the matrix in Figure \ref{fig:matrix1sided}.
There are $mn$ large cells, each of them contains an array of size $m$.
The large cells correspond to the positions of the man and the woman.
The arrays correspond to the position of the dog. Consider an
optimal ``walk'' of the gang that ends in the position $(a_{i},a_{p},b_{j},b_{q})$
(the black circle). The previous position of the gang correspond to a vertex
that can be located only in one of the 4 cells  $(a_{i},b_{j})$,$(a_{i-1},b_{j})$,$(a_{i},b_{j-1})$,$(a_{i-1},b_{j-1})$.
Moreover, it can only be one of the vertices marked in orange or the
red circles. If, for example, it is located in the left top orange
area, then $C(a_{i},a_{p},b_{j})=A_{i-1,j}[p-1]$, because $A_{i-1,j}[p-1]$
is the minimum number of steps of the dog when position of the man
and the woman is $(a_{i-1},b_{j})$. If it is the left top black circle,
then it is simply $C(a_{i-1},a_{p},b_{j})$ since the dog stayed at
the same position. Symmetrically, this is true for the other 3 large cells.
}

\begin{algorithm}[H]
\begin{description}
\item [{for}] $i=1$ to $m$
\begin{description}
\item [{for}] $j=1$ to $n$
\begin{description}
\item [{for}] $p=1$ to $m$

$
\begin{aligned}
&X_{(-1,0)}&=&\min\begin{cases}
A_{i-1,j}[p-1]+1\\
X(a_{i-1},a_{p},b_{j})
\end{cases}\\
&X_{(0,-1)}&=&\min\begin{cases}
A_{i,j-1}[p-1]+1\\
X(a_{i},a_{p},b_{j-1})
\end{cases}\\
&X_{(-1,-1)}&=&\min\begin{cases}
A_{i-1,j-1}[p-1]+1\\
X(a_{i-1},a_{p},b_{j-1})
\end{cases}\\
&X_{(0,0)}&=&A_{i,j}[p-1]+1\\
&X(a_{i},a_{p},b_{j})&=&\min\{X_{(-1,0)},X_{(0,-1)},X_{(-1,-1)},X_{(0,0)}\}\\
&A_{i,j}[p]&=&\min\{A_{i,j}[p-1], X(a_{i},a_{p},b_{j})\}\\
\end{aligned}
$

\end{description}
\end{description}
\item [{return}] $\underset{p}{\min}\{X(a_{m},a_{p},b_{n})\}$
\end{description}
\caption{1-sided CPS-3F using dynamic programming}
\end{algorithm}

\begin{theorem}
The 1-sided chain pair simplification problem under the discrete \frechet\ distance can be solved in $O(m^{2}n)$ time.
\end{theorem}

We now study the natural problem that is obtained from the 1-sided chain pair simplification problem by omitting the requirement that $\dfd(A,A') \le \delta_1$.

\begin{problem}[Relaxed 1-Sided Chain Pair Simplification]\hfill \\
    \textbf{Instance:} Given a pair of polygonal chains $A$ and $B$ of lengths $m$ and $n$,
    respectively, an integer $k$, and a real number $\delta>0$.\\
    \textbf{Problem:} Does there exist a chain $A'$ of at most $k$ vertices,
    such that the vertices of $A'$ are from $A$ and $\dfd(A',B) \le \delta$?
\end{problem}

This problem induces two optimization problems (as in \cite{BeregJWYZ08}), depending on whether we wish to optimize the length of $A'$ or the
distance between $A'$ and $B$. Below we solve both of them, beginning with the former problem.

\subsection{Minimizing \texorpdfstring{$k$}{k} given \texorpdfstring{$\delta$}{delta}}

In this problem, we wish to minimize the length of $A'$ without exceeding the allowed error bound.
\begin{problem}\label{prb:min-num-fitting}
Given two chains $A=(a_{1},\ldots,a_{m})$ and $B=(b_{1},\ldots,b_{n})$
and an error bound $\delta>0$, find a simplification $A'$ of $A$ of minimum length,
such that the vertices of $A'$ are from $A$ and $\dfd(A',B)\le\delta$.
\end{problem}


For $B=A$, Bereg et al.~\cite{BeregJWYZ08} presented an $O(n^2)$-time dynamic programming algorithm.
(For the case where the vertices of $A'$ are not necessarily from $A$, they presented an $O(n\log n)$-time greedy algorithm.)

\begin{theorem}
\label{thm:min-num-fitting}Problem~\ref{prb:min-num-fitting} can be solved in $O(mn)$ time and space.
\end{theorem}
\begin{proof}

We present an $O(mn)$-time dynamic programming algorithm. The algorithm finds the length of an optimal simplification;
the actual simplification is constructed by backtracking the algorithm's actions.

Define two $m \times n$ tables, $O$ and $X$.
The cell $O[i,j]$ will store the length of a minimum-length simplification $A^{i}$
of $A[i \ldots m]$ that begins at $a_{i}$ and such that $\dfd(A^{i},B[j \ldots n])\le\delta$.
The algorithm will return the value $\min_{\substack{{1\le i\le m}}} O[i,1]$.

We use the table $X$ to assist us in the computation of $O$. More precisely, we define:
$$
X[i,j]=\min_{\substack{{i'\ge i}}} O[i',j]\, .
$$
Notice that $X[i,j]$ is simply the minimum of $X[i+1,j]$ and $O[i,j]$.

We compute $O[-,-]$ and $X[-,-]$ simultaneously, where the outer for-loop is governed by (decreasing) $i$ and
the inner for-loop by (decreasing) $j$.
First, notice that if $d(a_{i},b_{j})>\delta$, then there is no simplification fulfilling the required conditions, so we set $O[i,j]=\infty$.
Second, the entries (in both tables) where $i=m$ or $j=n$ can be handled easily.
In general, if $d(a_{i},b_{j})\le\delta$, we set
$$
O[i,j]=\min\{O[i,j+1],\, X[i+1,j+1]+1\}\, .
$$

We now justify this setting.
Let $A^{i}$ be a minimum-length
simplification of $A[i \ldots n]$ that begins at $a_{i}$ and such that
$\dfd(A^{i},B[j \ldots n])\le\delta$. The initial configuration of the joint walk along $A^{i}$ and $B[j \ldots n]$ is $(a_i,b_j)$.
The next configuration is either $(a_i,b_{j+1})$, $(a_{i'},b_{j})$ for some $i'\ge i+1$, or $(a_{i'},b_{j+1})$ for some $i'\ge i+1$.
However, clearly $X[i+1,j+1] \le X[i+1,j]$, so we may disregard the middle option.

\old{Let $W^{i}$ be a paired
walk of $A^{i}$ and $B[j...n]$ with cost $\delta^{i}$. Notice that
for every pair $(A_{k}^{i},B_{k})\in W^{i}$ it holds that $|A_{k}^{i}|=1$,
otherwise we could remove points from $A^{i}$ and get a smaller simplification.
Thus we get that $A_{1}^{i}=\{a_{i}\}$ and $b_{j}\in B_{1}$. Clearly,
if $b_{j+1}\in B_{1}$ then $\delta^{i}=O[i,j+1]$. Else, let $i_{2}$
be the index such that $A_{2}^{i}=\{a_{i_{2}}\}$ let $i'$ be the
index such that $X[i,j]=O[i',j+1]$, so $O[i',j+1]\le O[i_{2},j+1]$,
and because $A^{i}$ is of minimum length we have $O[i',j+1]+1=O[i_{2},j+1]+1=\delta^{i}$.
}

\end{proof}

\subsection{Minimizing \texorpdfstring{$\delta$}{delta} given \texorpdfstring{$k$}{k}}

In this problem, we wish to minimize the discrete \frechet\ distance between $A'$ and $B$, without exceeding the allowed length.
\begin{problem} \label{prb:min-delta-fitting}
 Given two chains $A=(a_{1},\ldots,a_{m})$ and $B=(b_{1},\ldots,b_{n})$
and a positive integer $k$, find a simplification $A'$ of $A$ of length at most $k$,
such that the vertices of $A'$ are from $A$ and $\df(A',B)$ is minimized.
\end{problem}


For $B=A$, Bereg et al.~\cite{BeregJWYZ08} presented an $O(n^3)$-time dynamic programming algorithm.
(For the case where the vertices of $A'$ are not necessarily from $A$, they presented an $O(kn\log n\log(n/k))$-time greedy algorithm.)
We give an $O(mn\log{(mn)})$-time algorithm for our problem, which yields an $O(n^2 \log n)$-time algorithm for $B=A$, thus significantly improving the result of Bereg et al.

\begin{theorem}
Problem \ref{prb:min-delta-fitting} can be solved in $O(mn\log{(mn)})$ time and $O(mn)$ space.
\end{theorem}
\begin{proof}
Set $D=\{d(a,b) | a \in A, b \in B\}$. Then, clearly, $\df(A',B) \in D$, for any simplification $A'$ of $A$.
Thus, we can perform a binary search over $D$ for an optimal simplification of length at most $k$.
Given $\delta \in D$, we apply the algorithm for Problem~\ref{prb:min-num-fitting} to find (in $O(mn)$ time) a simplification $A'$ of $A$
of minimum length such that $\df(A',B) \le \delta$.
Now, if $|A'|>k$, then we proceed to try a larger bound, and if $|A'| \le k$,
then we proceed to try a smaller bound. After $O(\log{(mn)})$ iterations we reach the optimal bound.
\end{proof}

\section{Some Empirical Results} \label{sec:empirical}

In this section, we show some empirical results obtained by running
a C++ implementation of Algorithm~\ref{alg:cps3fdp} on a standard desktop machine.
The best available algorithm prior to this work was Algorithm~FIND-CPS3F$^+$, i.e., the algorithm (mentioned in the introduction)
for the optimization version of CPS-3F$^+$, proposed by Wylie and Zhu~\cite{WylieZ13}.
This algorithm is a 2-approximation algorithm for the optimization version of CPS-3F~\cite{WylieZ13}, and,
obviously, it cannot outperform Algorithm~\ref{alg:cps3fdp} in the length of the simplification that it computes.
The goal of this experimental study is thus twofold: (i) to verify that Algorithm~\ref{alg:cps3fdp} can cope with real datasets taken from
the Protein Data Bank (PDB), and (ii) to examine the actual improvement obtained in the length of the simplification
w.r.t. Algorithm~FIND-CPS3F$^+$.

Our results (summarized below) show that indeed Algorithm~\ref{alg:cps3fdp} can handle real datasets.
(This is not true for the initial algorithm, i.e., Algorithm~\ref{alg:cps3f}, whose $O(m^3n^3)$ space requirement would lead
to memory overflow for most proteins. Recall that there may be as many as 500$\sim$600 $\alpha$-carbon atoms along a protein backbone.)
Moreover, our results show significant improvement in the length of the simplification, which is very important for the underlying structural biology
applications.

As in~\cite{WylieZ13}, we consider two cases: similar chain length and
varying chain length comparisons. We use the same data and parameters
as in~\cite{WylieZ13}.

\subsection{Similar chain length comparisons}

We use the same seven pairs of protein backbones from the Protein Data Bank
which were used in~\cite{WylieZ13}.
To be consistent, we use the same sets of $\delta_1,\delta_2,\delta_3$ (in {\aa}ngstr\"{o}ms --- note that
the distance between two consecutive nodes, or $\alpha$-carbon atoms, on a protein backbone,
is typically between 3.7 to 3.8 {\aa}ngstr\"{o}ms).
The results are summarized in Tables~\ref{Tab:compare1}-\ref{Tab:compare3}.

\begin{table}[hbp]
\centering
\begin{tabular}{|p{1.5cm}|p{1cm}|p{1cm}|p{1cm}|p{1cm}|p{2.5cm}|p{2.5cm}|}
 \hline
Protein Chain(B)     & $|B|$   & $\delta_1$ & $\delta_2$ & $\delta_3$  & $\max\{|A''|,|B''|\}$ by CPS-3F$^+$ \cite{WylieZ13}  & $\max\{|A'|,|B'|\}$ by CPS-3F \\

\hline
1hfj.c  & 325  & 4 & 4 & 1 & 109 & 83 \\
\hline
1qd1.b   & 325  & 4  & 4  & 21 & 126 & 82\\
\hline
1toh  & 325  & 4  & 4 &21   & 149 & 84 \\
\hline
4eca.c   &325 & 4 & 4  &6   & 111 & 83  \\
\hline
1d9q.d & 297  & 4  & 4  & 20  & 130  & 82\\
\hline
4cea.b & 325  & 4  & 4  & 5  & 111  & 82\\
\hline
4cea.d & 325  & 4  & 4  & 5   & 113  & 84\\
\hline
\end{tabular}
\caption{Comparison of Algorithm FIND-CPS-3F$^{+}$~\cite{WylieZ13} and Algorithm~\ref{alg:cps3fdp} in this paper with 107j.a (Chain A) of Length 325.
Here $A''$ and $B''$ are the chains simplified from $A$ and $B$, respectively, using the former (approximation) algorithm FIND-CPS-3F$^{+}$, and
$A'$ and $B'$ are the chains simplified from $A$ and $B$, respectively, using Algorithm~\ref{alg:cps3fdp}.}
\label{Tab:compare1}
\end{table}

In Table~\ref{Tab:compare1},
$\delta_3$ is set to $\lceil\dfd(A,B)\rceil$.
From this table one can see that
with Algorithm~\ref{alg:cps3fdp}, we get $\max\{A,B\}/\max\{|A'|,|B'|\}\approx 4$, while with Algorithm FIND-CPS-3F$^{+}$,
we get $\max\{A,B\}/\max\{|A''|,|B''|\}\approx 3$, using the same data and parameters.
Hence, $\max\{|A'|,|B'|\}/\max\{|A''|,|B''|\}\approx 3/4$.

\begin{table}[!htp]
\centering
\begin{tabular}{|p{1.5cm}|p{1cm}|p{1cm}|p{1cm}|p{1cm}|p{2.5cm}|p{2.5cm}|}
 \hline
Protein Chain(B)     & $|B|$   & $\delta_1$ & $\delta_2$ & $\delta_3$  & $\max\{|A''|,|B''|\}$ by CPS-3F$^+$ \cite{WylieZ13} & $\max\{|A'|,|B'|\}$ by CPS-3F \\

\hline
1hfj.c  & 325  & 12 & 12 & 1 & 26 & 15 \\
\hline
1qd1.b   & 325  & 15  & 15  & 12 & 21 & 11\\
\hline
1toh  & 325  & 16  & 16 & 13   & 22 & 11 \\
\hline
4eca.c   &325 & 12 & 12  & 3   & 27 & 16  \\
\hline
1d9q.d & 297  & 15  & 15  & 13  & 24  & 12\\
\hline
4cea.b & 325  & 12  & 12  & 3  & 26  & 15\\
\hline
4cea.d & 325  & 12  & 12  & 3   & 32  & 16\\
\hline
\end{tabular}
 \caption{Comparison of Algorithm FIND-CPS-3F$^{+}$ \cite{WylieZ13} and Algorithm~\ref{alg:cps3fdp} in this paper with 107j.a (Chain A) of Length 325.}
\label{Tab:compare2}
\end{table}

In Table~\ref{Tab:compare2}, the parameters $\delta_1=\delta_2$ are
set to much larger values than in Table~\ref{Tab:compare1} (allowing us to set $\delta_3$ to smaller values).
The exact solutions by Algorithm~\ref{alg:cps3fdp} are even better now (w.r.t. Algorithm~FIND-CPS-3F$^{+}$).
From Table~\ref{Tab:compare2}, one can see that $\max\{|A'|,|B'|\}/\max\{|A''|,|B''|\}\approx 1/2$.

\subsection{Varying chain length comparisons}

In Table~\ref{Tab:compare3}, we simplify $A$ with several $B$ chains of
varying lengths. The parameter $\delta_1$ is not set to be equal
to $\delta_2$ anymore. From the table, it can be seen that
$\max\{|A'|,|B'|\}/\max\{|A''|,|B''|\}$ are mostly bounded by $2/3$ to $1/2$.

\begin{table}[!hbp]
\centering
\begin{tabular}{|p{1.5cm}|p{1cm}|p{1cm}|p{1cm}|p{1cm}|p{2.5cm}|p{2.5cm}|}
 \hline
Protein Chain(B)     & $|B|$   & $\delta_1$ & $\delta_2$ & $\delta_3$  & $\max\{|A''|,|B''|\}$ by CPS-3F$^+$ \cite{WylieZ13} & $\max\{|A'|,|B'|\}$ by CPS-3F \\

\hline
3ntx.a  & 322  & 10 & 10 & 5 & 39 & 25 \\
\hline
1wls.a   & 316  & 15  & 13  & 6 & 22 & 14\\
\hline
2eq5.a  & 215  & 8  & 6 & 19   & 58 & 32 \\
\hline
2zsk.a   &219 & 12 & 8  & 17   & 38 &  19 \\
\hline
1zq1.a & 418  & 10  & 12  & 19  & 45  & 23\\
\hline
3jq0.a & 457  & 12  & 12  & 26  & 70  & 36\\
\hline
2fep.a & 273  & 12  & 12  & 10   & 11  & 6\\
\hline
\end{tabular}
 \caption{Comparison of Algorithm FIND-CPS-3F$^{+}$~\cite{WylieZ13} and Algorithm~\ref{alg:cps3fdp} in this paper with 107j.a (Chain A) of Length 325. Here chain $B$ is of varying lengths.}
\label{Tab:compare3}
\end{table}

\old{
We comment that the space complexity of the initial algorithm for the minimization version of CPS-3F (Algorithm~\ref{alg:cps3f}), i.e.,
$O(m^3n^3)$, is not efficient enough for most proteins and
would lead to memory overflow even when $m,n$ are about 100 --- the upper bound
of protein backbone sizes are roughly 600. The improved version, i.e.,
CPS-3F using dynamic programming (Algorithm~2), handles all these scenarios.
}

\begin{remark}
Computing $A',B'$ for a pair of protein backbones $A,B$ might take several hours.
For example, for the largest pair (i.e., $A=$107j.a and $B=$3jq0.a in
Table~\ref{Tab:compare3}) it takes about 20 hours, so finding heuristics for
expediting the computation would be desirable.
One such heuristic, is to run Algorithm~FIND-CPS-3F$^{+}$~\cite{WylieZ13} to obtain a smaller
upper bound on $r$, i.e., $\max\{|A''|,|B''|\}$ instead of $m$, before running
Algorithm~\ref{alg:cps3fdp}.
\end{remark}

\old{
\section{Conclusion} \label{sec:conclusion}
In this paper we introduced the weighted chain pair simplification problem (WCPS-3F) and proved that
it is weakly \npc. We then resolved the complexity of CPS-3F, where all weights
are equal to one, which has been an open problem since 2008.
Specifically, we presented a polynomial-time solution for CPS-3F, and
showed how it can be extended to a pseudo-polynomial time
solution for weighted CPS-3F.

However, the running time of our algorithm is extremely high. Since the CPS-3F
problem has several applications that requires an efficient running time, an open question
is whether it is possible to reduce the running time for the CPS-3F problem. Notice that
in our solution we don't rely on the fact that the vertices of the simplification are
from the chain itself, a detail that may help reduce the running time.

We know that MIN-CPS-3F has at least two 2-approximations \cite{Wylie13,WylieZ13}, but
it remains open whether a constant factor approximation exists for MIN-WCPS-3F, and what the lower
bound on the approximation factor might be.

A lot of research for the \frechet\ distance is focused on planar applications, so an important
open question is whether the running time for CPS-3F can be reduced, especially if the two curves are planar (in 2D).

We presented a special case of the CPS problem which has a quadratic solution. The one-sided and two-sided shortcuts with the discrete \frechet\ distance \cite{AvrahamFKKS14} are special cases of CPS-3F (in the plane) that are subquadratic. What other special cases of CPS-3F, or WCPS-3F, are within reach with restrictions based on dimensionality, types of curves, or the amount of simplification allowed?

In our problem we require that the vertices of the simplification are from simplified chain, since arbitrary vertices are not meaningful in the context of comparing backbones of proteins, which was the main motivation for the CPS problem. But a version with arbitrary vertices should be examined too.
}

\bibliographystyle{alpha}
\bibliography{refs}

\end{document}